\documentclass[10pt, conference, letterpaper]{IEEEtran}
\IEEEoverridecommandlockouts

\def\BibTeX{{\rm B\kern-.05em{\sc i\kern-.025em b}\kern-.08em
    T\kern-.1667em\lower.7ex\hbox{E}\kern-.125emX}}
\ifCLASSINFOpdf
\fi
\usepackage{verbatim}
\usepackage{algorithm}
\usepackage{graphicx}
\usepackage{cite}
\usepackage{bm}
\usepackage{amssymb,amsfonts}
\usepackage{algorithmicx}
\usepackage[noend]{algpseudocode}
\usepackage{setspace}
\usepackage{float}
\usepackage{subfigure}
\usepackage{epstopdf}
\usepackage{amsthm}
\newtheorem{lemma}{Lemma}

\newtheorem{theorem}{Theorem}

\usepackage{stfloats}
\usepackage{enumerate}
\usepackage{multirow}
\usepackage{latexsym}
\usepackage{multicol,lipsum}
\usepackage{cuted} 

\usepackage{tikz}
\usepackage{textcomp}
\usepackage{xcolor}
\usepackage{cancel}

\pdfoutput=1
\usepackage[cmex10]{amsmath}
\usepackage[cmex10]{mathtools}

\algdef{SE}[DOWHILE]{Do}{doWhile}{\algorithmicdo}[1]{\algorithmicwhile\ #1}%
\allowdisplaybreaks
\makeatletter 
\def\@eqnnum{{\normalsize \normalcolor (\theequation)}} 

\makeatother 
\begin{document}
\title{Pinching Antenna-Assisted ISAC with Waveguide Mode Selection}
\author{Ruotong Zhao, Yijia Zhang, Shaokang Hu, and Derrick Wing Kwan Ng,~\IEEEmembership{Fellow,~IEEE} \\
\thanks{This work is supported by the HKUST-UNSW Global Research Impact Program \& Global Knowledge Network Awards and the Commonwealth through an Australian Government Research Training Program Scholarship [DOI: https://doi.org/10.82133/C42F-K220].}
\IEEEauthorblockA{School of Electrical Engineering and Telecommunications, University of New South Wales, Sydney, NSW 2052, Australia
\vspace{-3mm}
}}

\maketitle

\begin{abstract}
Conventional pinching antenna (PA)-assisted integrated sensing and communication (ISAC) architectures typically assume static receiver locations or predetermined receive waveguides, thereby underutilizing the inherent spatial degrees of freedom.
This paper proposes a novel mode-selectable PA-assisted ISAC framework to maximize the post-combining sensing signal-to-noise ratio while satisfying multi-user quality-of-service constraints by jointly optimizing the waveguide mode selection, transmit beamforming, and transmit/receive PA positions. To tackle the resulting mixed-integer nonconvex optimization problem, we develop a low-complexity block-coordinate descent algorithm that leverages a penalty-based majorization-minimization method to achieve high-quality suboptimal solutions. Numerical results demonstrate that the proposed design significantly outperforms both traditional PA and fixed-antenna benchmarks by synergistically harnessing spatial adaptability and modal reconfigurability. In particular, the mode-selectable design enables the coordinated optimization of transmit/receive operations and sensing–communication resource allocation, thereby maintaining sensing robustness under stringent communication requirements.
\end{abstract}
\vspace{-0.8em}

\section{Introduction}

Integrated sensing and communication (ISAC) enables simultaneous sensing and communication within a unified transceiver architecture, improving spectral efficiency and reducing deployment costs~\cite {xu2022robust}. However, the sensing performance is highly sensitive to propagation geometry and blockages due to severe round-trip path loss~\cite{liu2021cramer}, which motivates flexible-antenna architectures capable of dynamically reshaping the effective channel for reliable communication and accurate sensing~\cite{wu2023movable}.
Pinching-antenna (PA) technology offers a practical realization of such flexibility by activating radiating points at controllable locations, unlocking additional spatial degrees of freedom (DoF) to effectively mitigate path loss and enhance link quality~\cite{conference}. The PA system (PASS) is particularly beneficial for geometry-dependent round-trip ISAC, and the near-field operating regime of PASS yields richer spatial signatures for high-resolution sensing, e.g., joint range-and-angle inference,~\cite{guo2025learning}.
Nevertheless, PA-assisted ISAC entails resource-allocation complexity due to the PA-location-dependent coupling in path loss and in-waveguide phase shifts, along with the competition for limited waveguide resources, which naturally rise to an inherent sensing--communication trade-off.
Recent works have begun to explore PASS-assisted ISAC, including dedicated transmit/receive waveguide architectures and pinching-based beamforming optimization~\cite{11197530}, learning-enabled joint PA placement and beamforming design~\cite{qin2025joint, guo2025learning}, multi-waveguide designs with explicit sensing signal-to-noise ratio (SNR) constraints~\cite{mao2025multi}, and Cramér–Rao bound (CRB)-oriented localization-driven designs~\cite{li2025pinching}.

Although several recent works have advanced PA-assisted ISAC architectures~\cite{11197530, qin2025joint, mao2025multi, guo2025learning, li2025pinching}, most existing designs primarily focus on downlink transmission and assume a fixed sensing receiver (either co-located with the base station (BS) or attached to a pre-assigned receive waveguide). Such rigid receiver configurations limit sensing performance and system reconfigurability, and can ultimately underutilize scarce waveguide resources. Since the number and placement of waveguides are intrinsically constrained, waveguide assignment, i.e., how to configure each available waveguide for transmission or reception, becomes a critical yet largely overlooked design dimension.
Motivated by these identified research gaps and the need to effectively exploit the available spatial DoF, we investigate a PA‑assisted ISAC system in which each waveguide is configured for either downlink transmission or echo reception via waveguide‑mode selection. Under practical communication quality-of-service (QoS) constraints, we maximize the post-combining sensing SNR at the BS by jointly optimizing waveguide-mode assignments, transmit beamforming, and the transmit/receive PA positions.

The main contributions are summarized as:
i) We propose a mode-selectable waveguide-enabled PA-assisted ISAC framework capable of simultaneously serving multiple communication users and sensing a target by explicitly accounting for both the transmission and echo-reception processes.
ii) We formulate the design as a mixed-integer nonconvex resource-allocation problem and develop a computationally-efficient block-coordinate descent (BCD)-based algorithm that alternates between $(a)$ beamforming and waveguide-mode selection and (b) transmit/receive PA positioning; each block is handled via a penalty-based majorization-minimization (MM) procedure to address the remaining nonconvexities.
iii) Numerical results demonstrate that the proposed joint optimization of waveguide-mode selection and PA placement substantially enhances sensing performance while satisfying communication QoS requirements. This validates the effectiveness of exploiting the additional spatial DoF enabled by PA reconfigurability and adaptive waveguide utilization.

\textbf{Notation:} 
Boldface capital and lowercase letters denote matrices and vectors, respectively. $\textbf{A}^\top$, $\textbf{A}^H$, $\text{Tr}(\textbf{A})$, and $\text{Rank}(\textbf{A})$ represent the transpose, Hermitian transpose, trace, and rank of $\textbf{A}$, respectively. 
The operator $\mathrm{diag}(\textbf{a})$ forms a diagonal matrix from vector $\textbf{a}$, while $\mathrm{Diag}(\textbf{A})$ extracts the main diagonal of matrix $\textbf{A}$.
$\textbf{I}_N$, $\emptyset$,  $\mathbb{C}^{N \times M}$ ($\mathbb{R}^{N \times M}$), and $\mathbb{H}^{N}$ denote the identity matrix, an empty set, the spaces of $N \times M$ complex (real), and $N \times N$ Hermitian matrices, respectively. The largest eigenvalue of $\textbf{A}$ and its associated eigenvector are $\lambda_{\rm max}\hspace{-1mm}\left(\hspace{-0.2mm}\textbf{A}\hspace{-0.2mm}\right)$ and $\bm{\lambda}_{\rm max}\hspace{-1mm}\left(\hspace{-0.2mm}\textbf{A}\hspace{-0.2mm}\right)$, respectively. We utilize $|\cdot|$, $\mathbb{E}\{\cdot\}$, $\textbf{1}$, $j = \sqrt{-1}$,  $\Re(a)$, and $\Im(a)$ to represent the absolute value, statistical expectation, the all-one vector, the imaginary unit, the real part, and the imaginary part of the complex number $a$, respectively. The circularly symmetric complex Gaussian (CSCG) distribution is denoted by $\mathcal{CN}(\mu, \sigma^2)$, with mean $\mu$ and variance $\sigma^2$.

\section{System Model and Problem Statement}\label{sec:System}
\begin{figure}[t]
\centerline{\includegraphics[width=2.8in]{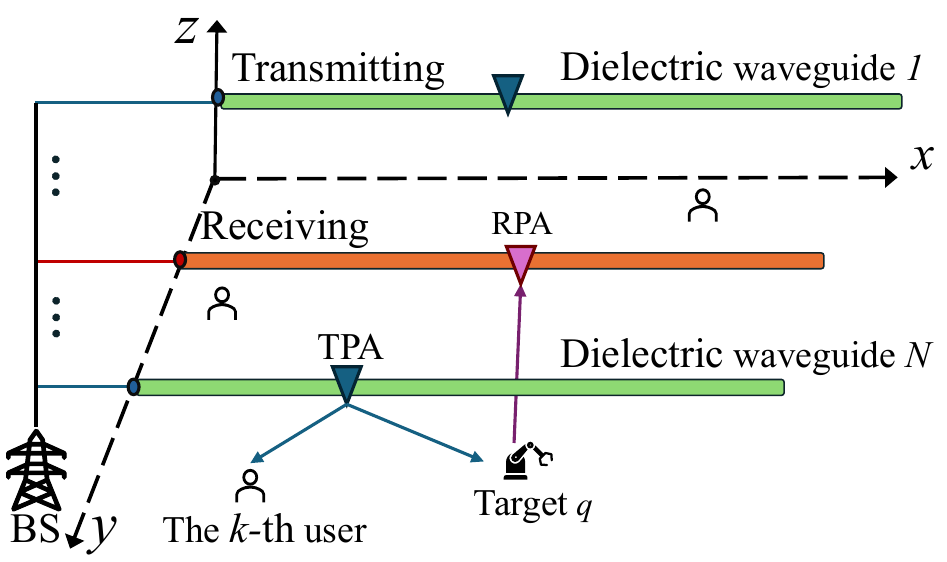}}
\caption{A PA-assisted ISAC system with multiple users and a single target.}
\vspace{-1.2em}
\label{fig:SYSTEM}
\end{figure}
As illustrated in Fig.~1, we consider a PA-assisted ISAC system, in which a dual-functional BS mounted at a height of $d$ meters above the ground that feeds $N$ dielectric waveguides of length $L$ meters. Each waveguide operates either in transmission mode or reception mode, modeled by a binary mode-selection vector $\bm{\tau}\triangleq[\tau_1,\ldots,\tau_N]^{\top}$ with $\tau_n\in\{0,1\}$, $\forall n \in \mathcal{N} \triangleq \{1, ..., N\}$. In particular, $\tau_n=1$ indicates that the $n$-th waveguide operates as a transmitting waveguide (TWG) with the corresponding index set $\mathcal{N}_{\rm t}\triangleq\{n\in\mathcal{N}\mid \tau_n=1\}$, while $\tau_n=0$ indicates that it is a receiving waveguide (RWG) with $\mathcal{N}_{\rm r}\triangleq\{n\in\mathcal{N}\mid \tau_n=0\}$, satisfying $\mathcal{N}_{\rm t}\cup\mathcal{N}_{\rm r}=\mathcal{N}$ and $\mathcal{N}_{\rm t}\cap\mathcal{N}_{\rm r}=\emptyset$. Moreover, we define the diagonal selection matrices for TWGs and RWGs as $\textbf{T}\triangleq\text{diag}(\bm{\tau})$ and $\textbf{R}\triangleq \text{diag}(1-\tau_1, ..., 1-\tau_N)$, respectively. Since the BS simultaneously serves $K$ single-antenna users, denoted as $\mathcal{K} \triangleq \{1, ..., K\}$, and performs radar sensing toward a single stationary point target under a narrowband LoS propagation model~\cite{mao2025multi}, we impose the constraint $K \leq \textbf{1}^{\top}\bm{\tau} \leq N-1$.

We assume that each waveguide can accommodate at most one transmitting pinching antenna (TPA) or one receiving pinching antenna (RPA)\footnote{We activate at most one PA per RWG to avoid intra-waveguide re-radiation coupling, which would otherwise yield a coupled and analytically intractable receive model~\cite{guo2025learning}. However, the extension to multiple simultaneously activated TPAs on a TWG is straightforward, as demonstrated in~\cite{conference, blockage}.}.
By adopting a three-dimensional (3D) Cartesian coordinate system, the positions of the feed point, TPA, and RPA of the $n$-th waveguide, $\forall n$, are denoted by
$\bm{\psi}_n^{\text{FP}} = \left[0,D_n,d\right]^{\top}$, $\bm{\psi}_{n}^{\rm TPA} = \left[x^{\rm TPA}_{n}, D_n, d\right]^{\top}$, and $\bm{\psi}_{n}^{\rm RPA} = \left[x^{\rm RPA}_{n}, D_n, d\right]^{\top}$, respectively, where $D_n$ is the coordinate along the $y$-axis for the $n$-th waveguide.
Furthermore, we denote the $k$-th user and the target by $u_k, \forall k$, and $q$, respectively, with corresponding 3D positions given by $\bm{\phi}_k = \left[x_k, y_k, 0\right]^{\top}$ and $\bm{\phi}_q = \left[x_q, y_q, 0\right]^{\top}$.

\subsection{Channel and Signal Model}
To investigate the performance upper bound, we neglect the intra-waveguide power attenuation~\cite{guo2025learning, li2025pinching} and adopt a quasi-static channel model with perfect\footnote{The sensing target parameters (e.g., range/angle) and channel estimation, along with the induced imperfections, are left for future work.} channel state information and assume that the target location is known at the BS~\cite{hu2025}. The downlink spherical-wave near-field wireless channel and the in-waveguide phase-shift coefficient between the PA of the $n$-th waveguide, $\forall n$, and the $k$-th user, $\forall k \in 
\mathcal{K}$, are~\cite{conference}:
\begin{align}\label{eq:channel_user}
h_{n,u_k}^{\rm t} &= \frac{\eta^{1/2} e^{-j\frac{2\pi}{\lambda}\|\bm{\phi}_k-\bm{\psi}^{\rm TPA}_{n}\|}}
     {\|\bm{\phi}_k-\bm{\psi}^{\rm TPA}_{n}\|}, \nonumber \\
g_n^{\rm t} &=\hspace{-0.8mm} e^{-j\frac{2\pi}{\lambda_g}\hspace{-0.5mm}\left\|\hspace{-0.4mm}\bm{\psi}_n^{\rm{FP}} \hspace{-0.8mm}-\hspace{-0.2mm} \bm{\psi}_{n}^{\rm TPA}\hspace{-0.6mm}\right\|}, \hspace{0.2mm} 
\end{align}
respectively, where 
$\eta = \tfrac{c^2}{16\pi^2f_c^2}$ and $\lambda_g = \frac{\lambda}{\eta_{\text{eff}}}$ with $c, f_c$, $\lambda$ and $\eta_{\text{eff}} > 1$ representing the speed of light, carrier frequency, wavelength, and effective refractive index~\cite{hu2025}, respectively. 
The equivalent downlink channel of the $k$-th user is defined as $\bm{\beta}_{u_k} \hspace{-1mm}= \Bigl[ 
h_{1,u_k}^{\rm t} g_1^{\rm t}, \ldots, h_{N,u_k}^{\rm t} g_N^{\rm t}
\Bigr]^{\hspace{-0.5mm}\top}\hspace{-2mm}, \forall k$, and analogously the effective downlink channel vector of the target is denoted by
$\bm{\beta}_{q}
\hspace{-1mm}= \hspace{-1mm}\Bigl[
h_{1,q}^{\rm t}g_1^{\rm t}, \ldots,  h_{N,q}^{\rm t} g_N^{\rm t}
\Bigr]^{\top}$.
Similarly, we define the composite received channel vector from the target across the $N$ waveguides as $\textbf{c}_{\rm R} = [h^{\rm r}_{1,q}g^{\rm r}_1, ..., h^{\rm r}_{N,q}g^{\rm r}_N]^{\top} \in \mathbb{C}^{N \times 1}$.

We let $s_k\in\mathbb{C}$ denote the information symbol intended for $u_k$ and define the transmitted signal vector as $\textbf{s}\triangleq \sum_{k\in\mathcal{K}} \textbf{w}_k s_k \in \mathbb{C}^{N\times 1}$, where $\mathbb{E}\{|s_k|^2\} \hspace{-1mm}=\hspace{-1mm} 1$, $\mathbb{E}\{s_j^{*}s_k\}\hspace{-1mm}=\hspace{-1mm} 0$, $\forall j \hspace{-1mm}\neq\hspace{-1mm} k$, and $\textbf{w}_k \hspace{-1mm}=\hspace{-1mm} [w_{k,1},\ldots,w_{k,N}]^{\top} \hspace{-2mm}\in\hspace{-1mm} \mathbb{C}^{N\times 1}$ is the beamforming vector. Following~\cite{attiah2024beamforming}, we reuse the downlink communication waveform for sensing\footnote{\hspace{-0.5mm}A dedicated sensing waveform could expand the DoFs but increase complexity for optimization and receiver processing. We reserve this investigation for future work. Moreover, direct TPA-RPA leakage is assumed negligible via cancellation; target-induced downlink scattering and user-induced reflections are ignored as their magnitudes are much weaker than direct links~\cite{attiah2024beamforming, guo2025learning}.}. To enforce zero excitation on RWGs, we set their power budgets to zero via $\sum_{k \in \mathcal{K}}|w_{k,n}|^2 \le \tau_n P_n^{\max}$, where $P_n^{\max}$ denotes the power budget of the $n$-th waveguide.
Let $\alpha \in \mathbb{C}$ denote the radar cross section dependent target reflection coefficient~\cite{guo2025learning}, the received signal at the $k$-th user and the received echo vector at the BS are given by:
\begin{align}\label{eq:rx_sig}
r_k &= \bm{\beta}_{u_k}^H\textbf{s} + \epsilon_k, \forall k, \nonumber \\
\hat{\textbf{r}}_{\rm R} &= \alpha \textbf{R} \textbf{c}_{\rm R}\bm{\beta}_q^H
\textbf{s} + \bm{\epsilon}_{\rm R},
\end{align}
respectively. Here, $\epsilon_k \sim \mathcal{CN}\left(0, \sigma_{k}^2\right)$ and $\bm{\epsilon}_{\rm R} \sim \mathcal{CN}\left(\bm{0}, \sigma^2_{\rm R}\textbf{R}\right)$ denote the additive white Gaussian noise (AWGN) at the $k$-th user and the RPAs, respectively. 

\subsection{Problem Formulation}
We evaluate the ISAC performance from both communication and sensing perspectives. The achievable data rate in bits/s/Hz for the $k$-th user can be written as:
\begin{align}\label{eq:rate}
R_k &= \log_2\left(1 + \Gamma_k\right), \forall k, \nonumber \\
\Gamma_k &= \frac{|\bm{\beta}_{u_k}^H\textbf{w}_k|^2}{\underset{i \neq k}{\sum}|\bm{\beta}_{u_k}^H\textbf{w}_i|^2 + \sigma_k^2}, \forall k.
\end{align}
We utilize the post-combining sensing SNR at the BS as a proxy metric for target estimation reliability~\cite{liu2021cramer}. With $\textbf{u}_{\rm R} = \tfrac{\textbf{Rc}_{\rm R}}{||\textbf{Rc}_{\rm R}||_2}$, the received radar signal and the average SNR\footnote{We employ a normalized maximum-ratio combining (MRC) receiver to maximize the sensing SNR. Without loss of generality, we set $|\alpha|=1$ as the non-designable reflection coefficient only scales the SNR.} are:
\begin{align}
r_{\rm R} &= \alpha \textbf{u}_{\rm R}^H \textbf{R} \textbf{c}_{\rm R}\bm{\beta}_q^H
\textbf{s} + \textbf{u}_{\rm R}^H\bm{\epsilon}_{\rm R}, \nonumber \\ \Gamma^{\text{s}}
&\triangleq
\frac{\left\| \textbf{R} \textbf{c}_{\rm R}\right\|_2^2 \underset{k \in \mathcal{K}}{\sum} \left|\bm{\beta}_q^H
\textbf{w}_k\right|^2}
{\sigma^2_{\rm R}},
\end{align}
respectively.
We maximize the sensing SNR under communication QoS constraints by jointly optimizing waveguide-mode selection, transmit beamforming, and PA positioning, as:
\begin{align}
&\underset{\textbf{w}_k, \ x^{\rm{TPA}}_n,\ x^{\rm{RPA}}_n, \ \bm{\tau}} {\rm{maximize}} \hspace{2mm} \Gamma^{\rm s}  \nonumber \\
\rm{s.t.}\ &({\rm C1})\hspace{-1mm}:\hspace{-0.5mm} 0 \leq x_{n}^{\rm{TPA}} \leq L, \forall n, \ ({\rm C2})\hspace{-1mm}:\hspace{-0.5mm} 0 \leq x_{n}^{\rm{RPA}} \leq L, \forall n, \nonumber \\
&({\rm C3})\hspace{-1mm}:\hspace{-0.7mm} \sum_{k \in \mathcal{K}}\hspace{-0.5mm}\|\textbf{w}_k\|_2^2 \hspace{-0.5mm}\leq\hspace{-0.5mm} P^{\max}\hspace{-0.5mm},\nonumber \\
&({\rm C4})\hspace{-1mm}:\hspace{-1mm} \sum_{k \in \mathcal{K}}\hspace{-0.6mm}|\textbf{w}_{k,n}|^2 \hspace{-0.75mm}\leq\hspace{-0.65mm} \tau_n P^{\max}_n\hspace{-0.5mm}, \forall n, \ ({\rm C5})\hspace{-1mm}:\hspace{-0.5mm} \tau_n \in \{0, 1\}, \forall n, \nonumber \\ 
&({\rm C6})\hspace{-1mm}:\hspace{-0.5mm} K \leq \textbf{1}^{\top}\bm{\tau} \leq N-1, \ ({\rm C7})\hspace{-1mm}:\hspace{-0.5mm} R_k \geq R^{\rm Req}_{{\rm min}, k}, \, \forall k.
\label{eq: formulation}
\end{align}
Constraints $({\rm C1})$ and $({\rm C2})$ confine the PA locations within their respective waveguides. Constraints $({\rm C3})$ and $({\rm C4})$ enforce the total power limit, $P^{\max}$, and the per‑waveguide power budget, $P^{\max}_n, \forall n$, respectively. Moreover, constraints $({\rm C5})$ and $({\rm C6})$ ensure valid binary waveguide-mode assignments to support multi-user downlink transmission and echo reception. Finally, constraint $({\rm C7})$ guarantees the minimum communication QoS with a prescribed threshold $R^{\rm Req}_{{\rm min}, k}, \forall k$.

\section{Optimization Solution}
Since problem~\eqref{eq: formulation} is highly nonconvex, we adopt a BCD framework by alternating between optimizing $\{\textbf{w}_k,\bm{\tau}\}$ and $\{x^{\text{RPA}}_n, x^{\text{TPA}}_n\}$ to acquire a high-quality suboptimal solution.

\subsection{Solution of $\{\textbf{w}_k, \bm{\tau}\}$}
With $\textbf{H}_k = \bm{\beta}_{u_k}\bm{\beta}_{u_k}^H,   \textbf{W}_k = \textbf{w}_k\textbf{w}_k^H, \forall k$, $\textbf{G} = \bm{\beta}_{q}\bm{\beta}_{q}^H$,  $\textbf{C} = \text{diag}(|\textbf{c}_{\text{R},1}|^2, ..., |\textbf{c}_{\text{R},N}|^2)$, and $\Gamma^{\rm Req}_{{\rm min},k} = 2^{R^{\rm Req}_{{\rm min},k}}-1$ for a given $\{x^{\text{RPA}}_n, x^{\text{TPA}}_n\}$, the subproblem of beamforming and mode selection is written as:
\begin{align}
&\underset{\textbf{W}_k, \ \bm{\tau}} {\rm{maximize}} \hspace{2mm}
\frac{\text{Tr}(\textbf{R}^2\textbf{C}) \sum_{k \in \mathcal{K}} \text{Tr}(\textbf{GW}_k)}
{\sigma^2_{\rm R}} \nonumber \\
\rm{s.t.}\ &({\rm C5}), ({\rm C6}), \nonumber \\
&({\rm C3})\hspace{-0.8mm}:\hspace{-0.8mm} \sum_{k \in \mathcal{K}}\hspace{-0.5mm}\text{Tr}(\textbf{W}_k) \hspace{-0.7mm}\leq\hspace{-0.7mm} P^{\max}, \nonumber \\
&({\rm C4})\hspace{-0.8mm}:\hspace{-0.8mm} \sum_{k \in \mathcal{K}}[\textbf{W}_{k}]_{n,n} \hspace{-0.7mm}\leq\hspace{-0.7mm} \tau_n P^{\max}_n, \forall n,\nonumber \\
& ({\rm C7})\hspace{-0.8mm}:\hspace{-0.8mm}  \text{Tr}(\textbf{H}_k\textbf{W}_k) \geq \Gamma^{\rm Req}_{{\rm min},k} \Bigl(\sum_{i \neq k}\text{Tr}(\textbf{H}_k\textbf{W}_i) + \sigma_k^2\Bigl), \forall k, \nonumber \\
&({\rm C8})\hspace{-0.8mm}:\hspace{-0.8mm} \textbf{W}_k \bm{\succeq 0}, \forall k, \,
({\rm C9})\hspace{-0.8mm}:\hspace{-0.8mm} \text{Rank}(\textbf{W}_k) \le 1, \forall k,
\label{eq: formulation1_1}
\end{align}
which remains nonconvex due to the multiplicative objective, the binary nature of $\bm{\tau}$, and the rank-one constraints.

We handle $({\rm C5})$ by expressing it in its equivalent forms:
\begin{align}
({\rm C5a})\hspace{-0.8mm}:\hspace{-0.6mm} \sum_{n \in \mathcal{N}} (\tau_n \hspace{-0.6mm}-\hspace{-0.6mm} \tau_n^2) \leq 0; \, 
({\rm C5b})\hspace{-0.8mm}:\hspace{-0.6mm} 0 \hspace{-0.6mm}\le\hspace{-0.6mm} \tau_n \le \hspace{-0.6mm}1, \forall n,
\end{align}
where $({\rm C5a})$ admits a difference of convex (D.C.) structure and $({\rm C5b})$ is affine.
We augment $({\rm C5a})$ with a penalty term and introduce slack optimization variables $u,v\ge 0$ to tackle the bilinear challenge in the objective. Using the identity $uv=\frac{(u+v)^2-u^2-v^2}{2}$, the penalized reformulation is:
\begin{align}
&\underset{\textbf{W}_k, \ \bm{\tau}, \ u, \, v} {\rm{maximize}} \hspace{2mm}
\frac{(u+v)^2 - u^2 -v^2}{2\sigma_{\rm R}^2} - \rho_1 \sum_{n \in \mathcal{N}}(\tau_n - \tau_n^2) \nonumber \\
\rm{s.t.}&
({\rm C3}), ({\rm C4}), ({\rm C5b}), ({\rm C6})- ({\rm C9}),\nonumber \\
&({\rm C10})\hspace{-0.9mm}:\hspace{-0.8mm} u \hspace{-0.8mm}\le\hspace{-1mm} \sum_{n \in \mathcal{N}} \hspace{-0.5mm} |\textbf{c}_{{\rm R},n}|^2 (\hspace{-0.1mm}1\hspace{-0.2mm}-\hspace{-0.2mm}\tau_n\hspace{-0.1mm})^2\hspace{-0.2mm}, \nonumber \\
&({\rm C11})\hspace{-0.9mm}:\hspace{-0.8mm} v\hspace{-0.5mm} \le\hspace{-0.8mm} \sum_{k\in K}\hspace{-0.5mm}\text{Tr}(\textbf{GW}_k).
\label{eq: formulation1_2}
\end{align}
It can be shown that a sufficiently large penalty factor $\rho_1$ makes~\eqref{eq: formulation1_2} equivalent to~\eqref{eq: formulation1_1}~\cite{blockage}. To handle the nonconvexities, we apply an iterative MM technique via first-order Taylor expansions, deriving global lower bounds for the nonconvex terms and a convex subset of $({\rm C10})$, as:
\begin{align}
(\hspace{-0.4mm}u \hspace{-0.6mm} + \hspace{-0.6mm}v\hspace{-0.4mm})^2 & \ge ((\overline{u+v})^{2})^{(i_1)} \hspace{-1mm}=\hspace{-1mm} o^2 \hspace{-0.6mm}+\hspace{-0.6mm} 2o \hspace{-0.2mm} \bigl((u\hspace{-0.4mm}-\hspace{-0.4mm}u^{(i_1)}) \hspace{-0.6mm}+\hspace{-0.6mm} (v\hspace{-0.4mm}-\hspace{-0.4mm}v^{(i_1)})\bigl)\hspace{-0.2mm}, \nonumber\\
\tau_n^2 & \ge \overline{\tau_n^2}^{(i_1)} \nonumber = 2 \tau_n^{(i_1)} \tau_n-(\tau_n^{(i_1)})^2, \forall n, \\
(\overline{\rm C10}) &: u - \sum_{n \in \mathcal{N}} |\textbf{c}_{{\rm R},n}|^2 \left(1-2\tau_n+ \overline{\tau_n^2}^{(i_1)} \right) \le 0,
\end{align}
respectively, where $o\hspace{-1mm} =\hspace{-1mm} u^{(i_1)} \hspace{-1mm}+\hspace{-0.7mm}v^{(i_1)}$ and $(\overline{\rm C10}) \hspace{-0.7mm}\Rightarrow\hspace{-0.7mm} (\rm C10)$. The superscript ``$(i_1)$'' denotes the solutions obtained from the $i_1$-th MM iteration and a suboptimal solution to~\eqref{eq: formulation1_1} can be acquired by solving the below problem in the $(i_1+1)$-th MM iteration:
\begin{align}
&\underset{\textbf{W}_k, \ \bm{\tau}, \ u, \, v} {\rm{maximize}} \hspace{2mm}
\frac{(\overline{u+v})^{2(i_1)} \hspace{-0.5mm} - \hspace{-0.5mm} u^2 \hspace{-0.5mm} - \hspace{-0.5mm} v^2}{2\sigma_{\rm R}^2} - \rho_1 \hspace{-0.5mm} \sum_{n \in \mathcal{N}}(\tau_n - \overline{\tau_n^2}^{(i_1)}) \nonumber \\
\rm{s.t.} \,
&({\rm C3}), ({\rm C4}), ({\rm C5b}), ({\rm C6})-({\rm C8}), \nonumber \\
&\cancel{({\rm C9})}, (\overline{{\rm C10}}), ({\rm C11}), 
\label{eq: formulation1_3}
\end{align}
where the rank constraint~({\rm C9}) is relaxed such that problem~\eqref{eq: formulation1_3} is convex and can be solved by a standard convex programming solver. The tightness of the SDR is revealed in the following theorem.
\begin{theorem}
If $P_{\max},R_{{\rm min},k}^{\rm Req} > 0, \forall k$, and~\eqref{eq: formulation1_3} is feasible, then a solution satisfying $\text{Rank}(\textbf{W}_k) \le 1, \forall k$, can always be obtained from~\eqref{eq: formulation1_3}.
\end{theorem}
\begin{proof}
Due to space limitations, we provide only a proof sketch. By examining the Karush-Kuhn-Tucker conditions, it follows that there always exists an optimal rank-one solution $\textbf{W}_k^\star$. Moreover, $\textbf{W}_k^\star$ can be explicitly constructed from the optimal dual variables of the corresponding dual problem.
\end{proof}

\subsection{Solution of the PA positioning}
For a given $\textbf{w}_k$ and $ \bm{\tau}$ with binary entries, $(\textbf{R}(\bm{\tau}))^2=\textbf{R}(\bm{\tau})$. Hence, the PA-positioning subproblem can be written as:
\begin{align}
&\underset{x^{\rm{TPA}}_n,\ x^{\rm{RPA}}_n} {\rm{maximize}} \hspace{2mm} \frac{\text{Tr}(\textbf{R}\textbf{C})  \underset{k \in \mathcal{K}}{\sum} \text{Tr}(\textbf{GW}_k)}
{\sigma^2_{\rm R}} \nonumber \\
\rm{s.t.}\ &({\rm C1}), ({\rm C2}), ({\rm C7}).
\label{eq: formulation2}
\end{align}
Since $x_n^{\rm RPA}$ affects~\eqref{eq: formulation2} only through $\text{Tr}(\textbf{R}\textbf{C})$ via the distance-dependent path loss, we obtain the following result.
\begin{lemma}
For $x_q \in [0, L]$, and for each RWG, i.e., $[\textbf{R}]_{n,n}=1$, an optimal RPA location is $x_{n}^{{\rm RPA}\star}=x_q, \ \forall n$.
\end{lemma}
\begin{proof}
With
$
\text{Tr}(\textbf{R}\textbf{C})
=\sum_{n=1}^{N}\tfrac{[\textbf{R}]_{n,n}\eta}{(x_q-x_n^{\rm RPA})^2+\hat{s}_{n,q}}$ and $
\hat{s}_{n,q}\triangleq (y_q-D_n)^2+d^2$, each summand is strictly decreasing in $(x_q-x_n^{\rm RPA})^2$. Thus, the maximum is attained by minimizing $(x_q-x_n^{\rm RPA})^2$ over  $x_{n}^{{\rm RPA}}$, i.e., by the Euclidean projection onto $[0, L]$, which yields $x_{n}^{{\rm RPA}\star}=x_q$.
\end{proof}


We let $u^* = \text{Tr}(\textbf{RC}(x_n^{{\rm RPA}*}))$ be fixed, the entrywise positive vector $\textbf{f}_{k} \hspace{-1.2mm}=\hspace{-1.2mm} \left[\tfrac{1}{\left\|\bm{\phi}_k - \bm{\psi}^{\rm{TPA}}_{1}\right\|}, ...\right.$, $\left.\tfrac{1}{\left\|\bm{\phi}_k - \bm{\psi}^{\rm{TPA}}_{N}\right\|} \right]^T \hspace{-2mm}\in \mathbb{R}^{N \times 1} $, $\textbf{C}_{k} \hspace{-1.2mm}= \text{diag}\hspace{-1mm}\left(\hspace{-0.8mm}w_{k,1} \eta^{\frac{1}{2}}, ..., w_{k,N} \eta^{\frac{1}{2}}\hspace{-1mm}\right) \in \mathbb{C}^{N \times N}$, and $\textbf{a}_{k} \hspace{-1.7mm} =\hspace{-1.5mm}\left[ e^{-j \theta_{1,k}}, ...,  e^{-j \theta_{N,k}} \hspace{-0.7mm}\right]^T \in \mathbb{C}^{N \times 1}, \forall k$, where $\theta_{n, k} \hspace{-1mm} = \hspace{-1mm} \frac{2\pi}{\lambda}\left\| \bm{\phi}_k \hspace{-0.7mm}-\hspace{-0.7mm} \bm{\psi}^{\rm{TPA}}_{n}\right\| + {\frac{2\pi}{\lambda_g}\left\|\bm{\psi}_n^{\rm{FP}} \hspace{-1mm}-\hspace{-0.7mm} \bm{\psi}_{n}^{\rm{TPA}}\right\|}, \forall n,k$. With two slack variables $\textbf{A}_k = \textbf{a}_k \textbf{a}_k^H \in \mathbb{C}^{N \times N}$ and $\textbf{F}_k = \textbf{f}_k\textbf{f}_k^H \in \mathbb{R}^{N \times N}, \forall k$, $\text{Tr}(\textbf{A}_k\textbf{C}_k\textbf{F}_k\textbf{C}_k^H)$ = $\text{Tr}(\textbf{H}_k\textbf{W}_k)$. Slack variables $\textbf{A}_q$ and $\textbf{F}_q$ are defined analogously, $\text{Tr}(\textbf{A}_q\textbf{C}_k\textbf{F}_q\textbf{C}_k^H)$ = $\text{Tr}(\textbf{G}\textbf{W}_k)$.
Introducing a slack variable $\vartheta>0$ and exploiting the Frobenius-norm identity, problem~\eqref{eq: formulation2} is equivalently written as:
\begin{align}
&\underset{x^{\rm{TPA}}_n,\ \textbf{X}, \ \vartheta, \  \theta_{n,q}, \ \theta_{n,k}} {\rm{maximize}} \hspace{4mm} \frac{\vartheta u^*}
{\sigma^2_{\rm R}}  \nonumber \\
&  \rm{s.t.}\ ({\rm C1}), \nonumber \\
&({\rm C7})\hspace{-1mm}: \hspace{-0.7mm}  \tfrac{\Gamma^{\rm Req}_{{\rm min}, k}\hspace{-0.5mm}\sum_{i \neq k} \hspace{-0.6mm} \left( \left\| \hspace{-0.3mm} \textbf{F}_{\hspace{-0.5mm}k} \hspace{-0.3mm}+\hspace{-0.1mm} \textbf{C}_i^{\hspace{-0.2mm}H} \hspace{-0.9mm}\textbf{A}_k \hspace{-0.4mm} \textbf{C}_i  \hspace{-0.4mm}\right\|_{\hspace{-0.4mm}F}^2 
\hspace{-1mm}-\hspace{-0.1mm} \left\| \textbf{F}_{\hspace{-0.5mm}k} \right\|_{\hspace{-0.4mm}F}^2 
\hspace{-0.3mm}-\hspace{-0.2mm} \left\| \hspace{-0.4mm}\textbf{C}_i^{\hspace{-0.2mm}H} \hspace{-0.8mm} \textbf{A}_k \hspace{-0.3mm} \textbf{C}_i\hspace{-0.6mm} \right\|_{\hspace{-0.5mm}F}^2 \hspace{-0.4mm}\right) } {2}\hspace{-0.59mm}+\hspace{-0.59mm} \Gamma^{\rm Req}_{{\rm min},k} \sigma_k^2 \nonumber \\
& \hspace{7.8mm} \leq \hspace{-0.6mm} \tfrac{\left\| \textbf{F}_{\hspace{-0.3mm}k} \hspace{-0.5mm}+\hspace{-0.2mm} \textbf{C}_k^{\hspace{-0.2mm}H} \hspace{-0.5mm}\textbf{A}_k \hspace{-0.4mm} \textbf{C}_k  \right\|_F^2 
\hspace{-1.2mm}- \left\| \textbf{F}_{\hspace{-0.5mm}k} \right\|_{\hspace{-0.35mm}F}^2 
\hspace{-0.2mm}-\hspace{-0.2mm} \left\| \textbf{C}_k^{\hspace{-0.2mm}H} \hspace{-0.5mm} \textbf{A}_k \hspace{-0.3mm} \textbf{C}_k \right\|_F^2}{2}, \forall k, \nonumber \\
&({\rm C11})\hspace{-1mm}: \hspace{-0.7mm} \sum_{k \in \mathcal{K} }\Bigl(\left\| \textbf{F}_{\hspace{-0.5mm}q} \hspace{-0.8mm}+\hspace{-0.8mm} \textbf{C}_k^{\hspace{-0.2mm}H} \hspace{-0.99mm}\textbf{A}_q \hspace{-0.5mm} \textbf{C}_k  \right\|_F^2 
\hspace{-1.6mm}-\hspace{-0.7mm} \left\| \textbf{F}_{\hspace{-0.5mm}q} \right\|_F^2 
\hspace{-1mm}-\hspace{-0.7mm} \left\| \textbf{C}_k^{\hspace{-0.2mm}H} \hspace{-1.1mm} \textbf{A}_q \hspace{-0.3mm} \textbf{C}_k \right\|_F^2 \Bigl)\ge 2\vartheta, \nonumber \\
&({\rm C12a})\hspace{-1mm}: \hspace{-0.7mm} \left(x_{n}^{\rm TPA} - x_k\right)^2 \le \frac{1}{\text{Diag} (\textbf{F}_k)_{n}} - \hat{s}_{n,k}, \ \forall n, k, \nonumber \\
&({\rm C12b})\hspace{-1mm}: \hspace{-0.7mm} \left(x_{n}^{\rm TPA} - x_k\right)^2 \ge \frac{1}{\text{Diag} (\textbf{F}_k)_{n}} - \hat{s}_{n,k}, \ \forall n, k, \nonumber \\
&({\rm C13a})\hspace{-1mm}: \hspace{-0.7mm} \theta_{n,k} \hspace{-0.8mm} \le \hspace{-0.4mm} \frac{2\pi}{\lambda} \hspace{-0.49mm} \left(\hspace{-0.2mm}\text{Diag} (\textbf{F}_k)_{n}\hspace{-0.2mm}\right)^{\hspace{-0.3mm}-\frac{1}{2}} \hspace{-0.49mm}+ \hspace{-0.5mm}\frac{2 \pi}{\lambda_g}\hspace{-0.4mm} x_{n}^{\rm TPA}, \hspace{-0.3mm}\forall n, \hspace{-0.2mm}k, \nonumber \\
&({\rm C13b})\hspace{-1mm}: \hspace{-0.7mm} \theta_{n,k} \hspace{-0.8mm} \ge \hspace{-0.4mm} \frac{2\pi}{\lambda} \hspace{-0.49mm} \left(\text{Diag} (\textbf{F}_k)_{n}\right)^{\hspace{-0.2mm}-\frac{1}{2}} \hspace{-0.29mm}+ \hspace{-0.2mm}\frac{2 \pi}{\lambda_g}\hspace{-0.2mm} x_{n}^{\rm TPA}, \hspace{-0.2mm}\forall n, k, \nonumber \\
&({\rm C14})\hspace{-1mm}: \hspace{-0.7mm} [\hspace{-0.3mm}\textbf{A}_k\hspace{-0.3mm}]_{n,n'} \hspace{-1mm} = \hspace{-1mm} e^{-\hspace{-0.1mm}j\hspace{-0.1mm}\left(\hspace{-0.2mm} \theta_{n,k} -\theta_{n',k} \hspace{-0.3mm} \right)}, \forall n, n'\hspace{-0.35mm},  k, \nonumber \\
&({\rm C15}) \hspace{-1mm}: \hspace{-0.7mm}  \textbf{X} \bm{\succeq} \textbf{0},  \forall \, \textbf{X}, \ ({\rm C16})\hspace{-1mm}: \hspace{-0.7mm} \text{Rank}\left(\textbf{X}\right)\hspace{-0.45mm} =\hspace{-0.3mm} 1\hspace{-0.1mm},\hspace{-0.1mm} \forall \, \textbf{X}, \nonumber \\
&({\rm C17}), ({\rm C18}), ({\rm C19}), \nonumber \\
& ({\rm C20})\hspace{-1mm}: \hspace{-0.7mm} \text{diag}(\textbf{A}_k) = 1, \forall k, \text{diag}(\textbf{A}_q) = 1,
\label{eq: formulation2_1}
\end{align}
where  $\textbf{X} \hspace{-1mm}=\hspace{-1mm} \{\textbf{A}_k, \textbf{F}_k, \textbf{A}_q, \textbf{F}_q\}, \hat{s}_{n,k} \hspace{-1mm}\triangleq\hspace{-1mm} (y_k \hspace{-1mm}-\hspace{-1mm} D_n)^2\hspace{-1.5mm}+\hspace{-0.6mm} d^2$, and  $\text{Diag}(\textbf{F}_k)_{n}\hspace{-1mm}=\hspace{-2mm}[\textbf{F}_k]_{n,n}$. Constraints $({\rm C12})-({\rm C14})$ explicitly couple $x_n^{\rm TPA}$ to slack variables $\textbf{F}_k$, $\textbf{A}_k$, and $\theta_{n,k}$.
We omit the details of $({\rm C17})$-$ ({\rm C19})$, as they are constructed by replacing user index $k$ with $q$.
Specifically, the paired inequalities $({\rm C12a})-({\rm C12b})$ and $({\rm C13a})-({\rm C13b})$ enforce the desired equalities~\cite{blockage}, where $({\rm C13b})$ is convex and the others are D.C. structure.
To handle the rank-one constraints, we introduce the following lemma
\begin{lemma}
The rank-one constraints $({\rm C16})$ are equivalent to:
\end{lemma}
\vspace{-1.5em}
\begin{equation}
(\widetilde{{\rm C16}}): \ \left\|\textbf{X}\right\|_* - \left\|\textbf{X}\right\|_2 \leq 0, \textbf{X} \in \{\textbf{A}_k,\textbf{F}_k, \forall k, \textbf{A}_q, \textbf{F}_q\}.
\end{equation}

\begin{proof}
For any $\textbf{X} \in \mathbb{H}^n \succeq \textbf{0}$ , the inequality $\left\|\textbf{X}\right\|_* = \sum_i\Lambda_i \geq \left\|\textbf{X}\right\|_2 = \text{max}\{\Lambda_i\}$ holds, where $\Lambda_i$ is the $i$-th singular value of $\textbf{X}$. The equality holds if and only if $\textbf{X}$ is rank-one.
\end{proof}

Constraint $(\widetilde{{\rm C16}})$ now admits a D.C. form.
To tackle the challenging phase-coupling nonconvexity, we equivalently transform $({\rm C14})$ by leveraging the trigonometric decomposition~\cite{conference} and express it with the first column of $\textbf{A}_k$ as
\begin{align}
\label{eq: C14}
({\rm C14a})\hspace{-0.5mm} &:\hspace{-0.5mm} \Re([\textbf{A}_k]_{l,1}) \hspace{-0.3mm}=\hspace{-0.2mm} \cos(\hat{\theta}_{l,k}) , \nonumber \\
({\rm C14b})\hspace{-0.5mm}&:\hspace{--0.5mm} \Im([\textbf{A}_k]_{l,1}) \hspace{-0.3mm}=\hspace{-0.2mm} -\sin(\hat{\theta}_{l,k}),
\end{align}
where $l \in \{2, ..., N\}$ and $\hat{\theta}_{l,k} = \theta_{l,k} - \theta_{1,k}$. 
Constraint $({\rm C19})$ follows analogously and is omitted for brevity. We introduce a penalized form of the optimization problem \eqref{eq: formulation2_1} as:
\begin{align}\label{eq: formulation2_2}
&\underset{x^{\rm{TPA}}_n, \ \textbf{X}, \ \vartheta,\  \bm{\theta}} {\rm{maximize}} \hspace{2mm}
\tfrac{\vartheta u^*}
{\sigma^2_{\rm R}}- \rho_2 \mathcal{T}  \nonumber \\
\text{s.t.} & ({\rm C1}), ({\rm C7}), ({\rm C11}), ({\rm C12a})-({\rm C13b}),  \nonumber \\
&({\rm C15}),(\widetilde{\rm C16}), ({\rm C17a}) \hspace{-0.5mm}-\hspace{-0.5mm}({\rm C18b}), ({\rm C20}), \nonumber \\
&(\text C21)\hspace{-0.4mm}:\hspace{-0.4mm} \hat{\theta}_{l,k} \hspace{-0.5mm} =\hspace{-0.7mm} \theta_{l,k} \hspace{-0.7mm}-\hspace{-0.7mm} \theta_{1,k}, \forall k,l, \nonumber \\
&(\text C22)\hspace{-0.4mm}:\hspace{-0.4mm} \hat{\theta}_{l,q} \hspace{-0.5mm} =\hspace{-0.7mm} \theta_{l,q} \hspace{-0.7mm}-\hspace{-0.7mm} \theta_{1,q}, \forall l,
\end{align}
where $\mathcal{T} \hspace{-1mm}=\hspace{-1mm} \underset{k \in \mathcal{K}}{\sum} \overset{N}{\underset{l=2}{\sum}} 
\Phi\bigl([\textbf{A}_k]_{l,1}, \hat{\theta}_{l,k}\bigl) + \sum_{l=2}^{N}\Phi\bigl([\textbf{A}_q]_{l,1}, \hat{\theta}_{l,q}\bigl)$, $\bm{\theta} = \{\theta_{l,k}, \theta_{l,q}, \hat{\theta}_{l,k}, \hat{\theta}_{l,q}\}$ and { $\Phi\hspace{-0.4mm}\bigl(\hspace{-0.3mm}[\textbf{A}_k]_{\hspace{-0.3mm}l,\hspace{-0.3mm}1}\hspace{-0.3mm},\hat{\theta}_{\hspace{-0.3mm}l\hspace{-0.3mm},a}\hspace{-0.45mm}\bigl) \hspace{-0.35mm}=\hspace{-0.4mm} |\hspace{-0.3mm}\Re\hspace{-0.4mm}\bigl(\hspace{-0.2mm}[\textbf{A}_k]_{l,1}\hspace{-0.6mm}\bigl) \hspace{-0.4mm}$ $-$ $\hspace{-0.4mm} \cos\hspace{-0.4mm}\bigl(\hspace{-0.4mm}\hat{\theta}_{l,a}\hspace{-0.3mm}\bigl) \hspace{-0.3mm}|^2 \hspace{-1mm}+\hspace{-0.6mm} |\hspace{-0.3mm}\Im\hspace{-0.6mm}\left(\hspace{-0.2mm}[\hspace{-0.2mm}\textbf{A}_k]_{l,1}\hspace{-0.4mm}\right)\hspace{-0.6mm} + \hspace{-0.3mm}\sin\hspace{-0.3mm}\bigl(\hspace{-0.3mm}\hat{\theta}_{l,a}\hspace{-0.6mm}\bigl)\hspace{-0.5mm}|^2, \forall l, a \in \{k \in \mathcal{K}, q\}$}. Here, the constant $\rho_2 \gg 1$ is a penalty coefficient that discourages violations of the equality constraints $({\rm C14a}),({\rm C14b}), ({\rm C19a}),$ and $({\rm C19b})$, where problems~\eqref{eq: formulation2_1} and~\eqref{eq: formulation2_2} are equivalent~\cite{blockage}.

To address the nonconvexity raised by the D.C. constraints and the periodic penalty term, we employ a penalty-based MM procedure to construct global underestimators via first-order Taylor expansions, and the convex subsets of the D.C. constraint sets are given by: 
\begin{align}\label{eq: SCA}
({\widetilde{{\rm C7}}}): &\frac{\underset{i \neq k}{\sum} \Gamma_{\min, k}^{\rm Req}\left( \left\| \textbf{F}_{\hspace{-0.2mm}k} \hspace{-0.4mm}+ \textbf{C}_i^{\hspace{-0.2mm}H} \hspace{-0.7mm}\textbf{A}_k \hspace{-0.3mm} \textbf{C}_i\hspace{-0.3mm} \right\|_{\hspace{-0.3mm}F}^2 
\hspace{-0.4mm} -  \overline{\mathcal{S}_{k,i}^{\rm aff}}\right)}{2} + \Gamma_{\min, k}^{\rm Req} \sigma_k^2 \nonumber \\ 
&\leq \overline{|| \textbf{F}_k +  \textbf{C}_k^{\hspace{-0.2mm}H}\hspace{-0.6mm} \textbf{A}_{\hspace{-0.2mm}k} \hspace{-0.3mm} \textbf{C}_k ||_F^2} \hspace{-0.65mm}-\hspace{-0.55mm} \left\| \textbf{F} \right\|_F^2  \hspace{-0.65mm}-\hspace{-0.65mm} || \textbf{C}_k^{\hspace{-0.2mm}H}\hspace{-0.6mm} \textbf{A}_{\hspace{-0.2mm}k} \hspace{-0.3mm} \textbf{C}_k ||_F^2 \hspace{-0.25mm}, \forall k, \nonumber \\
(\widetilde{\rm C12a}): & \left(x_{n}^{\rm TPA} - x_k\right)^2 - z_{n,k}^{\rm aff} + \hat{s}_{n,k} \leq 0 , \forall n,k,  \nonumber \\  
(\widetilde{\rm C12b}): & \frac{1}{\text{Diag} (\textbf{F}_k)_{n}} - \hat{s}_{n,k} - x_{n,k}^{\rm aff} \leq 0, \forall n,k, \nonumber \\
(\widetilde{\rm C13a}): & \theta_{n,k} -\overline{z_{n,k}^{\rm aff}} - \frac{2 \pi}{\lambda_g} x_{n}^{\rm TPA} \leq 0, \ \forall n, k, \nonumber \\  
(\widetilde{\overline{{\rm C16}}}): & \left\|\textbf{X}_k\right\|_* - \left\|\textbf{X}_k^{\rm aff}\right\|_2-||\hspace{-0.4mm}\textbf{X}^{(\hspace{-0.3mm}i_2\hspace{-0.3mm})}\hspace{-0.6mm}||_{\hspace{-0.3mm}2} \leq 0, \ \forall k, \textbf{X},
\end{align}
where $\overline{\|\textbf F_{k} + \textbf C_{k}^H  \textbf A_k \textbf C_{k}\|_F^2} = 2 \text{Tr}\bigl(\mathcal S_{(i_2)}^{H}  \textbf C_{k}^{H}  \triangle\textbf A_k \textbf C_{k} \bigr)
        + 2\text{Tr}\bigl( \mathcal S_{(i_2)}^H  \triangle\textbf F_k \bigr)
       +\|\mathcal S_{(i_2)}\|_F^2$, $x_{n,k}^{\rm aff} = ( x_{n}^{\text{TPA}(i_2)} - x_k)^{ 2}  + 2 ( x_{n}^{\text{TPA}(i_2)}  -  x_k )  ( x_{n}^{\text{TPA}}  -  x_{n}^{\text{TPA}(i_2)} )$,  $z_{n,k}^{\rm aff}  =  \tfrac{1}{z_{n,k}^{(i_2)}} - \tfrac{z_{n,k} - z_{n,k}^{(i_2)}}{(z_{n,k}^{(i_2)})^2}$, $\overline{z_{n,k}^{\rm aff}} = \tfrac{2\pi} {\lambda}\left(z_{n,k}^{(i_2)}\right)^{-\frac{1}{2}} - \frac{\pi}{\lambda} \left(z_{n,k}^{(i_2)}\right)^{-\frac{3}{2}}\left(z- z_{n,k}^{(i_2)}\right)$,  $\overline{\mathcal{S}_{k,i}^{\rm aff}}  =  2\text{Tr}\bigr((\textbf{F}_{k}^{(i_{2})})^{H}\textbf{F}_{k}\bigr)- ||\textbf{C}_{i}^{H}\textbf{A}_{k}^{(i_{2})}\textbf{C}_{i}||_{F}^{2}- ||\textbf{F}_{k}^{(i_{2})}||_{F}^{2}+ 2 \text{Tr}\bigr((\textbf{C}_{i}  \textbf{C}_{i}^{H} \textbf{A}_{k}^{(i_{2})}\textbf{C}_{i}  \textbf{C}_{i}^{H})^{H}\textbf{A}_{k}\bigr)$, $\left\|\textbf{X}^{\rm aff}\right\|_{2}  =  \text{Tr}\left[\bm{\lambda}_{\rm max}(\textbf{X}^{(i_2)}) \bm{\lambda}^H_{\rm max}(\textbf{X}^{(i_2)}) (\textbf{X} - \textbf{X}^{(i_2)}) \right]$, $\mathcal{S}_{(i_2)} = \textbf{F}_k^{(i_2)} + \textbf{C}_k^H \textbf{A}_k^{(i_2)} \textbf{C}_k $, $\bm{\triangle} \textbf{F}_k = \textbf{F}_k - \textbf{F}_k^{(i_2)}$, $\text{Diag}(\textbf{F}_k)_{n} = z_{n,k}$, and  $\bm{\triangle} \textbf{A}_k = \textbf{A}_k - \textbf{A}_k^{(i_2)}$.
Note that $(\widetilde{{\rm C7}}), (\widetilde{{\rm C12a}}), (\widetilde{{\rm C12b}})$,$ (\widetilde{{\rm C13a}}), (\widetilde{\overline{\rm C16}}) \hspace{-0.6mm} \Rightarrow \hspace{-0.7mm}  ({\rm C7}), ({\rm C12a})$, $ ({\rm C12b}), ({\rm C13a}), (\widetilde{{\rm C16}})$, respectively. 
The terms with ``$(i_2)$'' are the solutions obtained in the $i_2$-th MM iteration. The convex constraints $(\widetilde{{\rm C11}})$ and $(\widetilde{{\rm C17a}})$--$(\widetilde{{\rm C18a}})$ can be derived analogously to~\eqref{eq: SCA} and therefore are omitted here.

Considering the challenging periodic terms, we apply the MM technique with a Lipschitz gradient surrogate to establish a valid global upper bound~\cite{mairal2015incremental} for $\Phi \hspace{-0.5mm}\bigr(\hspace{-0.35mm}[\textbf{A}_k]_{l,1}, \hat{\theta}_{l,k}\hspace{-0.5mm}\bigr) \hspace{-0.4mm}, \forall l, k$, as: 
\begin{align}
\overline{\Phi}^{(\hspace{-0.3mm}i_2\hspace{-0.3mm})}_{l,k} 
      & = 2\hspace{-0.4mm}\left(\hspace{-0.1mm}\hat{\Re}_{l,k}^{(\hspace{-0.3mm}i_2\hspace{-0.3mm})} \hspace{-0.2mm}\sin \hat{\theta}_{l,k}^{(\hspace{-0.3mm}i_2\hspace{-0.3mm})} \hspace{-0.8mm}+\hspace{-0.8mm} \hat{\Im}_{l,k}^{(\hspace{-0.3mm}i_2\hspace{-0.3mm})}\hspace{-0.7mm}\cos \hat{\theta}_{l,k}^{(\hspace{-0.3mm}i_2\hspace{-0.3mm})}\hspace{-0.3mm}\right)\hspace{-0.8mm}\triangle \hat{\theta}_{l,k} \nonumber \\
     &+ \frac{L_{\hspace{-0.2mm}\text{AR}} \hspace{-0.2mm} (\hspace{-0.4mm} \triangle \Re_{\hspace{-0.2mm}l\hspace{-0.2mm},\hspace{-0.2mm}k}\hspace{-0.4mm})^2 + L_\text{AI} (\triangle \Im_{l,k})^2 + L_{\rm TH} (\triangle \hat{\theta}_{l,k})^2}{2} \nonumber \\
      &+ 2\hat{\Re}_{l,k}^{(\hspace{-0.3mm}i_2\hspace{-0.3mm})}\hspace{-0.1mm}\triangle \Re_{l,k} \hspace{-0.2mm}+\hspace{-0.2mm}\Phi^{(i_2)} + \hspace{-0.4mm} 2 \hspace{-0.2mm} \hat{\Im}_{l,k}^{\hspace{-0.3mm}(\hspace{-0.3mm}i_2\hspace{-0.3mm})}\hspace{-0.2mm}\triangle \Im_{\hspace{-0.2mm}l\hspace{-0.2mm},k}, \forall k, l
\end{align}
where $\hat{\Re}_{l,k}^{(i_2)} = \Re([\textbf{A}_k]_{l,1}^{(i_2)}) -\hspace{-0.5mm} \cos \hat{\theta}_{l,k}^{(i_2)}$, $\hat{\Im}_{l,k}^{(i_2)} = \Im([\textbf{A}_k]_{l,1}^{(i_2)}) + \sin \hat{\theta}_{l,k}^{(i_2)}$, $\triangle \hat{\theta}_{l,k} = \hat{\theta}_{l,k} - \hat{\theta}_{l,k}^{(i_2)}$, $\triangle \Re_{l,k}=\Re([\textbf{A}_k]_{l,1}) - \Re([\textbf{A}_k]_{l,1}^{(i_2)}) $, $\triangle \Im_{l,k} = \Im([\textbf{A}_k]_{l,1}) - \Im([\textbf{A}_k]_{l,1}^{(i_2)})$, with the Lipschitz constants $L_{\rm AR} = L_{\rm AI} = 2$ and $L_{\rm TH} = 4$ based on the maximum curvature~\cite{mairal2015incremental}. 
The corresponding surrogate for $\Phi\bigl([\textbf{A}_q]_{l,1}, \hat{\theta}_{l,q}\bigl)$ is obtained analogously and is denoted by $\overline{\Phi}^{(\hspace{-0.3mm}i_2\hspace{-0.3mm})}_{l,q}$, $\forall l$.

Accordingly, the $(i_2 + 1)$-th iteration of the penalty-based MM method for this subproblem can be expressed as: 
\begin{align}
\label{eq: formulation2_3}
&\underset{x^{\rm{TPA}}_n, \ \textbf{X}, \, \vartheta,\  \bm{\theta}} {\rm{maximize}} \hspace{2mm}
\frac{\vartheta u^*}{\sigma_{\rm R}^2} - \rho_2 \left[\sum_{k=1}^{K} \sum_{l=2}^{N}
        \overline{\Phi}_{l,k}^{(i_2)} + \sum_{l=2}^{N}
        \overline{\Phi}_{l,q}^{(i_2)}\right] \nonumber \\
\rm{s.t.}\
&({\rm C1}), (\widetilde{\rm C7}), (\widetilde{\rm C11}),(\widetilde{\rm C12a})-(\widetilde{\rm C13a}),({\rm C13b}), ({\rm C15}),\nonumber \\
&(\widetilde{\overline{{\rm C16}}}), (\widetilde{\rm C17a})-(\widetilde{\rm C18a}), ({\rm C18b}), ({\rm C20})-({\rm C22}),
\end{align}
where problem~\eqref{eq: formulation2_3} is convex and yields a high-quality suboptimal solution to problem~\eqref{eq: formulation2}. Moreover, the proposed BCD algorithm is guaranteed to converge to a suboptimal solution of~\eqref{eq: formulation} with polynomial-time computational complexity~\cite{hu2025}.

\section{Simulation Results}
This section evaluates the proposed system via Monte-Carlo simulations with $d=3$~m and $f_c=12$~GHz~\cite{baduge2025frequency}. The user and target locations are independently and uniformly generated in a $20 \times 20$~$\text{m}^2$ area, and all results are averaged over multiple random realizations. Unless stated otherwise, we set $L = 20$ m, $D_n = \tfrac{L}{N}$ m, $\sigma_k = \sigma_{\rm R} = \sigma = -90$ dBm, $\forall k$, $N=8$, $K=3$, $P^{\max}=1$~W, $P^{\max}_n = \tfrac{P^{\max}}{N}$~W, $\forall n$, and $R_{\min,k}^{\rm Req}=R_{\min}=1$~bits/s/Hz, $\forall k$. We compare the proposed design with three baselines: i) \emph{Bench~1} (fixed waveguide mode), which follows the proposed optimization framework but leverages a random and fixed mode split with $N/2$ RWGs and $N/2$ TWGs; ii) \emph{Bench~2} (fixed RPAs), which constrains the RPAs to be located at the BS; and iii) \emph{Bench~3} (fixed antennas), which considers a conventional fixed-antenna array at the BS with optimized beamforming and antenna modes (transmission or reception) design.
\begin{figure}[t]
\centerline{\includegraphics[width=2.5in]{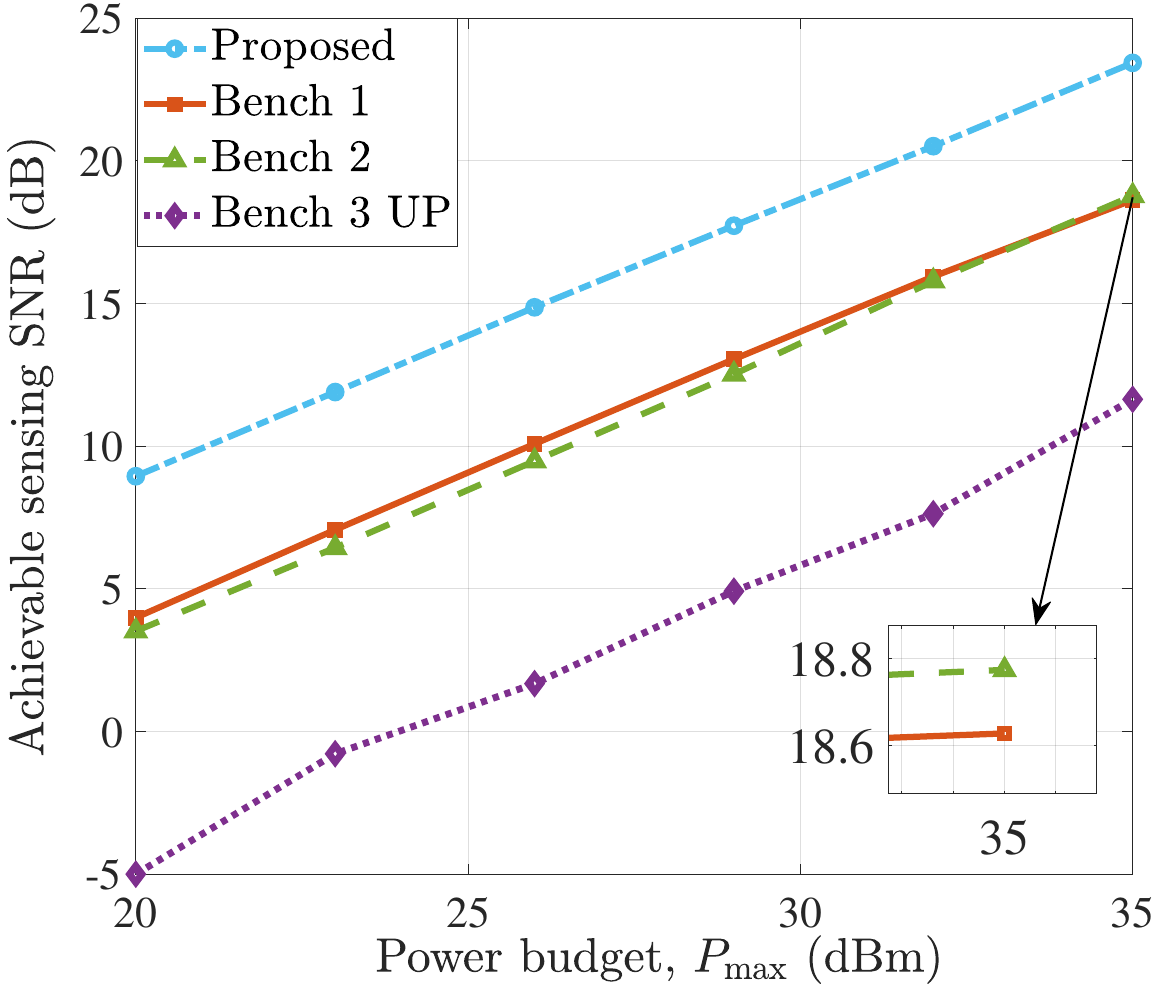}}
\vspace{-1mm}
\caption{Average post-combining sensing SNR versus $P_{\max}$ for different antenna systems.}
\label{fig: powerInc}
\vspace{-1em}
\end{figure}

Fig.~\ref{fig: powerInc} shows that the average round-trip sensing SNR of all schemes increases monotonically with the power budget $P_{\max}$, owing to the additional available energy for ISAC. The proposed design consistently achieves the highest sensing SNR over the entire power range, confirming that joint waveguide-mode selection and RPA positioning provide additional spatial DoF for near-field focusing and effectively mitigate the severe two-way propagation loss. 
Interestingly, \emph{Bench~1} slightly outperforms \emph{Bench~2} at low $P_{\max}$, whereas \emph{Bench~2} becomes marginally superior as $P_{\max}$ increases. This suggests that RPA positioning flexibility governs low-power sensing performance, whereas the fixed RWG/TWG allocation in \emph{Bench~1} becomes the dominant bottleneck at high $P_{\max}$. By contrast, the adaptive mode selection in \emph{Bench~2} coordinates the limited waveguide resources more effectively and hence exploits the additional power more efficiently for both transmission and reception. 
Since \emph{Bench~3} is generally infeasible under the QoS constraints, we instead evaluate its QoS-relaxed upper bound,  \emph{Bench~3 UP}. Nevertheless, its persistent performance gap compared to PA-assisted schemes underscores the inherent inadequacy of fixed-array architectures for near-field ISAC.

\begin{figure}[t]
\centerline{\includegraphics[width=2.5in]{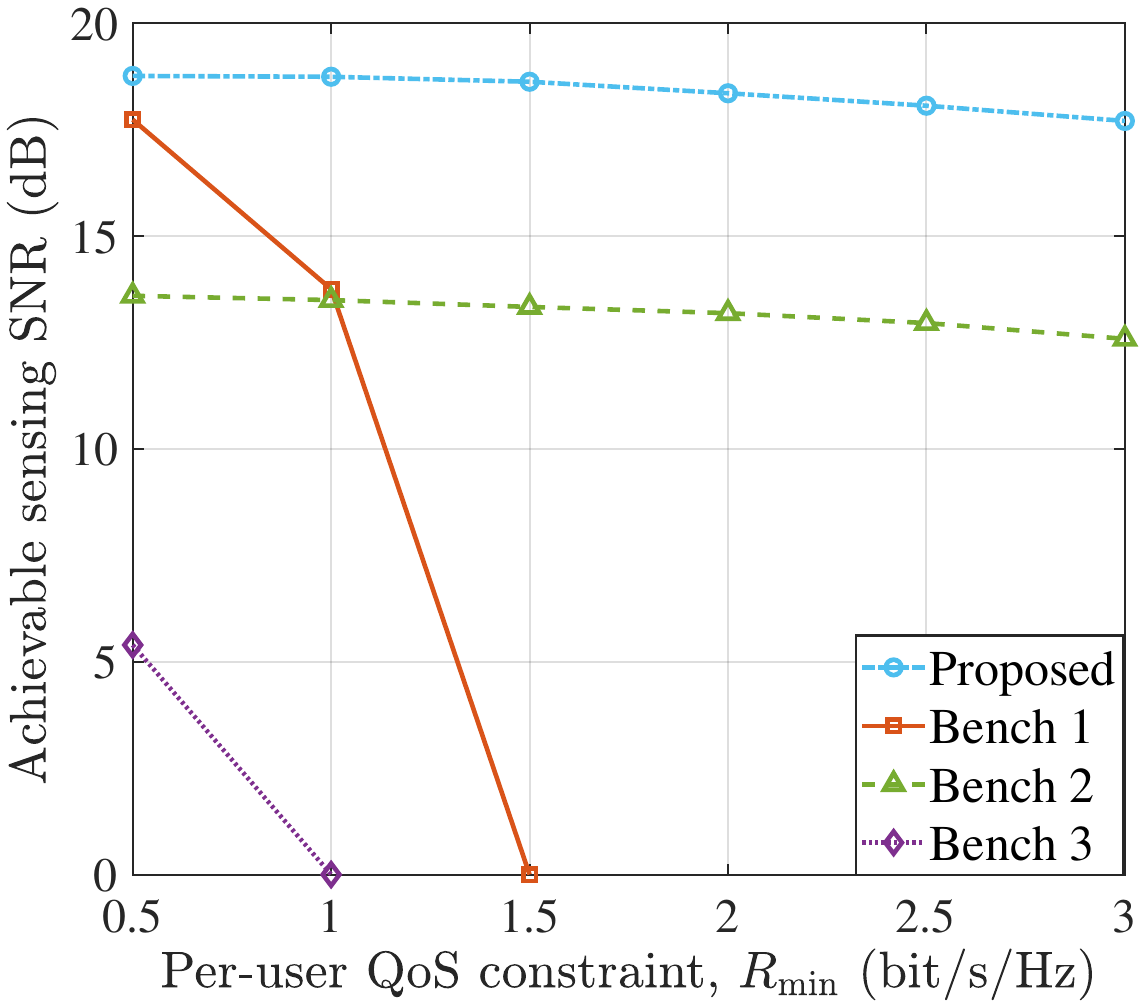}}
\vspace{-1mm}
\caption{Average post-combining sensing SNR versus the per-user QoS requirement $R_{\min}$ of different antenna systems.}
\label{fig: RminInc}
\vspace{-1.5em}
\end{figure}

Fig.~\ref{fig: RminInc} evaluates the average sensing SNR versus the QoS target $R_{\min}$, where infeasible realizations are mapped to zero under the adopted penalty rule. The figure reveals the intrinsic sensing--communication trade-off: as $R_{\min}$ increases, more available transmit power and spatial DoF are forced to be allocated to communication, thereby reducing the flexibility for sensing and hence limiting the sensing SNR. Notably, the proposed design experiences only mild performance degradation and remains feasible over the entire range, highlighting that the proposed joint design framework effectively enlarges the feasible solution set to improve both sensing and communication performance. In contrast, \emph{Bench~2} remains feasible but suffers a persistent SNR loss due to the additional path loss caused by fixed RPA locations. \emph{Bench~1} is competitive at low $R_{\min}$, yet degrades rapidly and eventually becomes infeasible as the fixed RWG/TWG split turns into a DoF bottleneck under stringent QoS constraints, indicating the importance of adaptive mode selection. Indeed, \emph{Bench~3} fails even earlier due to the lack of hardware reconfigurability.

\section{Conclusion}
This work proposed a mode-selectable waveguide-enabled PA-assisted ISAC framework, in which waveguides are adaptively configured for either transmission or echo reception. We maximized the post-combining sensing SNR under per-user QoS constraints by jointly optimizing waveguide modes, transmit beamforming, and the positions of the transmit/receive PAs. A low-complexity BCD method with penalty-based MM was developed to tackle the resulting mixed-integer nonconvex problem. Numerical results demonstrated consistent sensing gains and enhanced robustness to QoS tightening compared with fixed-mode/position baselines. More importantly, the results reveal a fundamental sensing--communication trade-off and demonstrate that the proposed adaptive waveguide-mode selection is key to coordinating the limited resources across transmission and reception, thereby yielding a more favorable balance between sensing robustness and communication performance under stringent QoS requirements.



\begin{thebibliography}{10}
\providecommand{\url}[1]{#1}
\csname url@samestyle\endcsname
\providecommand{\newblock}{\relax}
\providecommand{\bibinfo}[2]{#2}
\providecommand{\BIBentrySTDinterwordspacing}{\spaceskip=0pt\relax}
\providecommand{\BIBentryALTinterwordstretchfactor}{4}
\providecommand{\BIBentryALTinterwordspacing}{\spaceskip=\fontdimen2\font plus
\BIBentryALTinterwordstretchfactor\fontdimen3\font minus \fontdimen4\font\relax}
\providecommand{\BIBforeignlanguage}[2]{{%
\expandafter\ifx\csname l@#1\endcsname\relax
\typeout{** WARNING: IEEEtran.bst: No hyphenation pattern has been}%
\typeout{** loaded for the language `#1'. Using the pattern for}%
\typeout{** the default language instead.}%
\else
\language=\csname l@#1\endcsname
\fi
#2}}
\providecommand{\BIBdecl}{\relax}
\BIBdecl

\bibitem{xu2022robust}
D.~Xu, X.~Yu, D.~W.~K. Ng, A.~Schmeink, and R.~Schober, ``Robust and secure resource allocation for {ISAC} systems: A novel optimization framework for variable-length snapshots,'' \emph{IEEE Trans. Commun.}, vol.~70, no.~12, pp. 8196--8214, 2022.

\bibitem{liu2021cramer}
F.~Liu, Y.-F. Liu, A.~Li, C.~Masouros, and Y.~C. Eldar, ``Cram{\'e}r-{Rao} bound optimization for joint radar-communication beamforming,'' \emph{IEEE Trans. Signal Process.}, vol.~70, pp. 240--253, Dec. 2021.

\bibitem{wu2023movable}
Y.~Wu, D.~Xu, D.~W.~K. Ng, W.~Gerstacker, and R.~Schober, ``Movable antenna-enhanced multiuser communication: Jointly optimal discrete antenna positioning and beamforming,'' in \emph{Proc. IEEE Global Commun. Conf. (Globecom)}, Dec. 2023, pp. 7508--7513.

\bibitem{conference}
R.~Zhao, S.~Hu, D.~Mishra, and D.~W.~K. Ng, ``Resource allocation for multi-waveguide pinching antenna-assisted broadcast networks,'' in \emph{IEEE Global Commun. Conf. (Globecom) Wkshps.}, Dec. 2025, pp. 1--7.

\bibitem{guo2025learning}
J.~Guo, Y.~Liu, and A.~Nallanathan, ``Learning beamforming for pinching antenna system-enabled {ISAC} in low-altitude wireless networks,'' \emph{arXiv preprint arXiv:2512.04293}, 2025.

\bibitem{11197530}
Z.~Zhang, Z.~Wang, X.~Mu, B.~He, J.~Chen, and Y.~Liu, ``Integrated sensing and communications for pinching-antenna systems {(PASS)},'' \emph{IEEE Commun. Lett.}, vol.~29, no.~12, pp. 2929--2933, Dec. 2025.

\bibitem{qin2025joint}
Y.~Qin, Y.~Fu, and H.~Zhang, ``Joint antenna position and transmit power optimization for pinching antenna-assisted {ISAC} systems,'' \emph{IEEE Wireless Commun. Lett.}, vol.~14, no.~11, pp. 3535--3539, Nov. 2025.

\bibitem{mao2025multi}
W.~Mao, Y.~Lu, Y.~Xu, B.~Ai, O.~A. Dobre, and D.~Niyato, ``Multi-waveguide pinching antennas for {ISAC},'' \emph{IEEE Trans. Wireless Commun.}, vol.~25, pp. 5846--5858, Oct. 2025.

\bibitem{li2025pinching}
H.~Li, R.~Zhong, J.~Lei, and Y.~Liu, ``Pinching antenna system for integrated sensing and communications,'' \emph{IEEE Trans. wirelss Commun.}, vol.~25, pp. 13\,416--13\,429, Mar. 2026.

\bibitem{blockage}
R.~Zhao, S.~Hu, D.~Mishra, and D.~W.~K. Ng, ``Robust and secure blockage-aware pinching antenna-assisted wireless communication,'' \emph{IEEE Trans. Mob. Comput.}, pp. 1--18, early access, May 22, 2026, doi: 10.1109/TMC.2026.3695952.

\bibitem{hu2025}
S.~Hu, R.~Zhao, Y.~Liao, D.~W.~K. Ng, and J.~Yuan, ``Sum-rate maximization for pinching antenna-assisted {NOMA} systems with multiple dielectric waveguides,'' in \emph{IEEE Global Commun. Conf. (Globecom) Wkshps.}, Dec. 2025, pp. 1--7.

\bibitem{attiah2024beamforming}
K.~M. Attiah and W.~Yu, ``Beamforming design for integrated sensing and communications using uplink-downlink duality,'' in \emph{Proc. IEEE Int. Symp. Inf. Theory (ISIT)}, Jun. 2024, pp. 2808--2813.

\bibitem{mairal2015incremental}
J.~Mairal, ``Incremental majorization-minimization optimization with application to large-scale machine learning,'' \emph{SIAM J. Optim.}, vol.~25, no.~2, pp. 829--855, 2015.

\bibitem{baduge2025frequency}
G.~A. Baduge, M.~Vaezi, J.~K. Dassanayake, M.~Z. Hameed, E.~Ollila, and S.~A. Vorobyov, ``Frequency range 3 for {ISAC} in {6G}: Potentials and challenges,'' \emph{IEEE Commun. Mag.}, early access, May 13, 2026, doi: 10.1109/MCOM.001.2500382.

\end{thebibliography}
\end{document}